\documentclass[10pt,leqno]{amsart}
\usepackage{graphicx}
\usepackage{indentfirst,csquotes}

\usepackage{amssymb,amsthm,amsmath}
\usepackage{xcolor,paralist,hyperref,fancyhdr,etoolbox}
\usepackage[ruled,vlined,linesnumbered]{algorithm2e}

\newtheorem{theorem}{Theorem}[]
\newtheorem{definition}[theorem]{Definition}

\newtheorem{proposition}[theorem]{Proposition}
\newtheorem{corollary}[theorem]{Corollary}

\hypersetup{ colorlinks=true, linkcolor=black, filecolor=black, urlcolor=black }

\usepackage{lipsum}

\begin{document}
\title{Norm-Query Complexity of Algorithmic Problems in Finite-Dimensional p-adic Normed Spaces} 
\author{Zhefan Duan}
\author{Huawei Wu*}
\date{\today}
\begin{abstract}
We study the deterministic norm-query complexity of computational problems in finite-dimensional vector spaces over $\mathbb{Q}_p$ equipped with an arbitrary ultrametric norm. For orthogonalization, we prove that no uniform finite query bound depending only on the dimension exists: for every deterministic algorithm that produces an $N$-orthogonal basis for every ultrametric norm $N$, the number of norm queries is unbounded as $N$ varies. We then study the Longest Vector Problem (LVP) for a rank-$m$ $p$-adic lattice. By adapting a brute-force search to the general norm-query setting and eliminating the scalar redundancy among nonzero coefficient vectors modulo $p$, we obtain an algorithm using exactly $(p^m-1)/(p-1)$ norm queries for $m\ge 2$, and prove that no deterministic norm-query algorithm can use fewer queries in the worst case. Finally, we consider the Closest Vector Problem (CVP). Apart from the trivial cases in which no norm query is needed, we prove that the deterministic worst-case norm-query complexity of the CVP is unbounded, even when the lattice and the target vector are fixed.
\end{abstract} 
\maketitle

\begingroup
\renewcommand{\thefootnote}{}
\footnotetext[0]{*Corresponding author. E-mail: wuhuawei1996@gmail.com\\
\indent Zhefan Duan is with the Department of Mathematical Sciences, Tsinghua University, Beijing, China.\\
\indent Huawei Wu is an independent researcher.}
\endgroup

\bigskip

\section{Introduction}

Throughout this paper, $p$ denotes a prime number, $\mathbb{Q}_p$ the field of $p$-adic numbers, and $\mathbb{Z}_p$ the ring of $p$-adic integers. We write
\[
\mathbb{F}_p := \{0, 1, \dots, p-1\}
\]
for the finite field with $p$ elements. For $a\in\mathbb{Q}_p$, we write $|a|_p$ for the $p$-adic absolute value of $a$ and $v_p(a)$ for the $p$-adic valuation of $a$.

For a positive integer $r$, we denote by $\mathbb{P}^{r}(\mathbb{F}_p)$ the $r$-dimensional projective space over $\mathbb{F}_p$, viewed here simply as the set of one-dimensional $\mathbb{F}_p$-subspaces of $\mathbb{F}_p^{r+1}$. In particular,
\[
\#\mathbb{P}^{r}(\mathbb{F}_p)=1+p+\cdots+p^r=\frac{p^{r+1}-1}{p-1}.
\]

Let $V$ be a vector space over $\mathbb{Q}_p$. An ultrametric norm on $V$ is a function $N:V\to \mathbb{R}_{\ge 0}$ such that for all $x,y\in V$ and all $a\in\mathbb{Q}_p$, the following conditions hold:
\begin{enumerate}[(1)]
    \item $N(x)=0$ if and only if $x=0$;
    \item $N(ax)=|a|_p\,N(x)$;
    \item $N(x+y)\le \max\{N(x),N(y)\}$.
\end{enumerate}
Condition (3) is called the strong triangle inequality. Because of this property, ultrametric norms behave quite differently from the usual norms on $\mathbb{R}^n$ or $\mathbb{C}^n$. For example, if $N(x)\neq N(y)$, then
\[
N(x+y)=\max\{N(x),N(y)\}.
\]

In \cite{Weil1974}, Weil proved the following result.

\begin{proposition}[{\cite[Chapter II, Proposition 3]{Weil1974}}]
Let $V$ be an $n$-dimensional vector space over $\mathbb{Q}_p$, and let $N$ be an ultrametric norm on $V$. Then there exists a basis $\{\alpha_1,\alpha_2,\ldots,\alpha_n\}$ of $V$ such that for any $a_1,\cdots,a_n\in\mathbb{Q}_p$,
\[
N\Big(\sum_{i=1}^n a_i\alpha_i\Big)=\max_{1\le i\le n}\bigl(|a_i|_p\,N(\alpha_i)\bigr).
\]
\end{proposition}

Weil's theorem naturally motivates the following definition.

\begin{definition}
Let $V$ be an $n$-dimensional vector space over $\mathbb{Q}_p$, and let $N$ be an ultrametric norm on $V$. A basis $\{\alpha_1,\alpha_2,\ldots,\alpha_n\}$ of $V$ is called \emph{$N$-orthogonal} if for every $a_1,\cdots,a_n\in\mathbb{Q}_p$,
\[
N\Big(\sum_{i=1}^n a_i\alpha_i\Big)=\max_{1\le i\le n}\bigl(|a_i|_p\,N(\alpha_i)\bigr).
\]
\end{definition}

Weil's proof is non-constructive. From a computational point of view, it is natural to ask whether one can design an algorithm, analogous to the Gram--Schmidt orthogonalization process, which starts from an arbitrary basis and constructs an $N$-orthogonal basis for a general ultrametric norm $N$.

In \cite{deng2025p}, Deng transformed Weil's proof into a deterministic orthogonalization algorithm. This algorithm is closely related to several computational problems for $p$-adic lattices, which are also studied in this paper.

\begin{definition}[{\cite{deng2022some}}]
Let $V$ be an $n$-dimensional vector space over $\mathbb{Q}_p$, and let $\alpha_1,\cdots,\alpha_m$ ($1\le m\le n$) be $\mathbb{Q}_p$-linearly independent vectors in $V$. The \emph{$p$-adic lattice} generated by $\alpha_1,\cdots,\alpha_m$ is
\[
\mathcal{L}(\alpha_1,\cdots,\alpha_m)
=
\Bigl\{\sum_{i=1}^m a_i\alpha_i:\ a_i\in\mathbb{Z}_p\Bigr\}.
\]
The integers $m$ and $n$ are called the \emph{rank} and the \emph{ambient dimension} of the lattice, respectively. If $m=n$, then the lattice is called \emph{full-rank}.
\end{definition}

\begin{definition}[{\cite{deng2022some}}]
Let $V$ be an $n$-dimensional vector space over $\mathbb{Q}_p$, let $N$ be an ultrametric norm on $V$, and let $\mathcal{L}:=\mathcal{L}(\alpha_1,\cdots,\alpha_m)$ be a $p$-adic lattice in $V$ of rank $m$. Define a decreasing sequence of norm values associated to $\mathcal{L}$ by
\[
\lambda_1(\mathcal{L})=\max_{1\le i\le m}N(\alpha_i),
\]
and, for $j\ge 1$,
\[
\lambda_{j+1}(\mathcal{L})
=
\max\{N(v): v\in\mathcal{L},\ N(v)<\lambda_j(\mathcal{L})\}.
\]
The \emph{Longest Vector Problem (LVP)} is to find a vector $v\in\mathcal{L}$ such that
\[
N(v)=\lambda_2(\mathcal{L}).
\]
Let $t\in V$ be a target vector. The \emph{Closest Vector Problem (CVP)} is to find a vector $v\in\mathcal{L}$ such that $N(v-t)$ is minimized.
\end{definition}

Deng showed in \cite{deng2025p} that computing an $N$-orthogonal basis can be reduced to solving a sequence of CVPs. Moreover, in \cite{deng2022some}, Deng et al.\ gave a deterministic algorithm for solving the CVP, which in turn yields a deterministic orthogonalization algorithm.

To measure the complexity of such algorithms, we count only norm evaluations (equivalently, norm queries), i.e., evaluations of the form $N(x)$. It is worth noting that the number of norm computations required by their CVP algorithm depends not only on the dimension $n$ and rank $m$ of the $p$-adic lattice $\mathcal{L}$, but also on the norm $N(t)$ of the target vector $t$ and the maximal lattice norm $\lambda_1(\mathcal{L})$. Consequently, the complexity of their orthogonalization algorithm (measured by the number of norm computations) depends not only on the dimension $n$ of the ambient space $V$, but also on the norms of the vectors in the initial basis.

This is very different from the classical Gram--Schmidt orthogonalization process. Indeed, if one rewrites all inner-product computations in Gram--Schmidt in terms of norm computations (for example, via the polarization identity), then the total number of norm computations depends only on the dimension of the vector space and is independent of the specific norm.

This raises a natural question: does there exist an orthogonalization algorithm for ultrametric normed vector spaces whose number of norm computations depends only on the dimension $n$? In this paper, we prove that the answer is negative.

For the LVP, Deng et al.\ \cite{deng2022some} considered $p$-adic lattices in a finite extension of $\mathbb{Q}_p$ and gave a deterministic brute-force search algorithm using $O(p^m)$ $p$-adic absolute-value computations, where $m$ is the rank of the lattice. We first adapt their brute-force approach to the more general setting of a finite-dimensional $\mathbb{Q}_p$-vector space equipped with an arbitrary ultrametric norm. We then exploit the invariance of the norm under multiplication by $p$-adic units to remove the scalar redundancy in the search: nonzero coefficient vectors modulo $p$ that differ by a nonzero scalar need not be queried separately. This reduces the number of norm queries to
\[
\frac{p^m-1}{p-1}
\]
for $m\ge 2$. We further prove a matching lower bound, showing that this is the exact deterministic worst-case norm-query complexity of the LVP.

We finally consider the CVP. There are two trivial situations in which no norm query is needed: when the target vector belongs to the lattice, and when the lattice has rank one and the target lies in its $\mathbb{Q}_p$-span. We prove that, apart from these cases, no uniform finite bound exists for the number of norm queries required by a deterministic CVP algorithm. In particular, the worst-case norm-query complexity is unbounded even when the lattice and the target vector are fixed.

The proofs of our lower bounds are based on adversarial indistinguishability arguments. Together, these results illustrate that computational problems in ultrametric normed spaces can exhibit complexity behavior substantially different from their Euclidean counterparts, and suggest that the norm-query viewpoint may provide a useful framework for studying other computational problems in this setting.

\section{Norm-Query Complexity of Orthogonalization}

We first consider the norm-query complexity of orthogonalization. The following theorem shows that, in contrast to the classical Gram--Schmidt process, no uniform finite query bound depending only on the dimension exists in the ultrametric setting.

\begin{theorem}\label{thm:ortho}
    Let $V$ be an $n$-dimensional vector space over $\mathbb{Q}_p$ with $n\ge 2$, and let $\{\alpha_1,\cdots,\alpha_n\}$ be a basis of $V$. Assume that $\mathcal{A}$ is a deterministic orthogonalization algorithm which accesses the norm only through evaluations of the form $N(x)$, and such that, for any ultrametric norm $N$ on $V$, the algorithm $\mathcal{A}$ produces, starting from the basis $\{\alpha_1,\dots,\alpha_n\}$, an $N$-orthogonal basis in finitely many steps.
    
    For an ultrametric norm $N$ on $V$, let $g_{\mathcal{A}}(N)$ denote the number of norm computations required by $\mathcal{A}$ in the process of constructing an $N$-orthogonal basis. Then
    \[
    \sup_{N} g_{\mathcal{A}}(N)=\infty,
    \]
    where the supremum is taken over all ultrametric norms $N$ on $V$.
    \end{theorem}
    
    \begin{proof}
    We prove the theorem by contradiction. Assume that
    \[
    K:=\sup_{N} g_{\mathcal{A}}(N)<\infty.
    \]
    For some norms, the algorithm may terminate before the norm has been computed $K$ times. We can repeatedly compute $N(\alpha_1)$ to make up a total of $K$ evaluations, while keeping the original output unchanged. Therefore, we may assume without loss of generality that, for all norms, the algorithm computes the norm exactly $K$ times before it terminates. Moreover, since a basis is given, we may assume without loss of generality that $V=\mathbb{Q}_p^n$.

    For any $u=(u^{(1)},\cdots,u^{(n-1)})\in\mathbb{Q}_p^{n-1}$ and $M\in\mathbb{Z}$, define
    \[
    \ell_u(x)=x^{(n)}-\sum_{i=1}^{n-1}u^{(i)} x^{(i)}
    \quad\text{and}\quad
    N_{u,M}(x)=\max\Bigl(|\ell_u(x)|_p,\frac{\max_{1\le i\le n-1}|x^{(i)}|_p}{p^M}\Bigr)
    \]
    for $x=(x^{(1)},x^{(2)},\cdots,x^{(n)})\in\mathbb{Q}_p^n$. It is straightforward to verify that $N_{u,M}$ is an ultrametric norm on $\mathbb{Q}_p^n$. For any $c=(c^{(1)},c^{(2)},\cdots,c^{(n-1)})\in\mathbb{Q}_p^{n-1}$ and $T,L\in\mathbb{Z}$, define
    \[
    R(c,T,L):=\{(u,M)\in\mathbb{Q}_p^{n-1}\times\mathbb{Z}: u\equiv c\ (\mathrm{mod}\ p^T),\ M\ge L\},
    \]
    where $u\equiv c\ (\mathrm{mod}\ p^T)$ means that $u^{(i)}-c^{(i)}\in p^T\mathbb{Z}_p$ for all $1\le i\le n-1$.
    
    Let the vector used in the first norm query be
    \[
    x_1=(x_1^{(1)},x_1^{(2)},\cdots,x_1^{(n)}).
    \]
    Since this is the first norm query, the algorithm has not yet received any value of the norm, and hence $x_1$ is independent of the choice of the norm. We claim that there exist $c_1\in\mathbb{Q}_p^{n-1}$ and $T_1,L_1\in\mathbb{Z}$ such that for any $(u,M)\in R(c_1,T_1,L_1)$, we have
    \[
    N_{u,M}(x_1)=|\ell_{c_1}(x_1)|_p.
    \]
    
    There are two cases to consider:
    \begin{enumerate}
    \item If $x_1^{(i)}=0$ for all $1\le i\le n-1$, then for any $u\in\mathbb{Q}_p^{n-1}$ and $M\in\mathbb{Z}$,
    \[
    N_{u,M}(x_1)=|\ell_u(x_1)|_p=|x_1^{(n)}|_p.
    \]
    In this case, we can take $c_1=(0,0,\cdots,0)$ and $T_1=L_1=0$.
    
    \item Otherwise, we can choose $c_1\in\mathbb{Q}_p^{n-1}$ such that $\ell_{c_1}(x_1)\ne 0$. Put
    \[
    \gamma_1=v_p\bigl(\ell_{c_1}(x_1)\bigr)\in\mathbb{Z},
    \qquad
    \beta_1=\min_{1\le i\le n-1}v_p(x_1^{(i)})\in\mathbb{Z},
    \]
    and define
    \[
    T_1=\gamma_1-\beta_1+1,
    \qquad
    L_1=\gamma_1-\beta_1.
    \]
    For any $(u,M)\in R(c_1,T_1,L_1)$, consider the $p$-adic valuation of $\ell_u(x_1)$. Since
    \[
    \ell_u(x_1)=\ell_{c_1}(x_1)-\sum_{i=1}^{n-1}(u^{(i)}-c_1^{(i)})x_1^{(i)}
    \]
    and
    \begin{align*}
    v_p\Bigl(\sum_{i=1}^{n-1}(u^{(i)}-c_1^{(i)})x_1^{(i)}\Bigr)
    &\ge \min_{1\le i\le n-1}\Bigl(v_p(u^{(i)}-c_1^{(i)})+v_p(x_1^{(i)})\Bigr)\\
    &\ge \min_{1\le i\le n-1}v_p(u^{(i)}-c_1^{(i)})+\min_{1\le i\le n-1}v_p(x_1^{(i)})\\
    &\ge T_1+\beta_1=\gamma_1+1,
    \end{align*}
    we obtain
    \[
    v_p\bigl(\ell_u(x_1)\bigr)
    =\min\Bigl\{\gamma_1,\,v_p\Bigl(\sum_{i=1}^{n-1}(u^{(i)}-c_1^{(i)})x_1^{(i)}\Bigr)\Bigr\}
    =\gamma_1,
    \]
    and hence
    \[
    |\ell_u(x_1)|_p=p^{-\gamma_1}.
    \]
    Since $M\ge L_1$, we have
    \[
    \frac{\max_{1\le i\le n-1}|x_1^{(i)}|_p}{p^M}
    =p^{-\beta_1-M}
    \le p^{-\beta_1-L_1}
    =p^{-\gamma_1}.
    \]
    Therefore,
    \[
    N_{u,M}(x_1)
    =\max\Bigl\{|\ell_u(x_1)|_p,\frac{\max_{1\le i\le n-1}|x_1^{(i)}|_p}{p^M}\Bigr\}
    =p^{-\gamma_1}
    =|\ell_{c_1}(x_1)|_p.
    \]
    \end{enumerate}
    This proves the claim.
    
    We have shown that for every $(u,M)\in R(c_1,T_1,L_1)$, the first norm query returns the same real number $|\ell_{c_1}(x_1)|_p$. Since the algorithm $\mathcal{A}$ is deterministic, the vector used in the second norm query is therefore the same for all such $(u,M)$. Denote this vector by
    \[
    x_2=(x_2^{(1)},x_2^{(2)},\cdots,x_2^{(n)}).
    \]
    
    We claim that there exist $c_2\in\mathbb{Q}_p^{n-1}$ and $T_2,L_2\in\mathbb{Z}$ such that
    \[
    R(c_2,T_2,L_2)\subseteq R(c_1,T_1,L_1)
    \]
    and, for any $(u,M)\in R(c_2,T_2,L_2)$, we have
    \[
    N_{u,M}(x_2)=|\ell_{c_2}(x_2)|_p.
    \]
    
    Again there are two cases:
    \begin{enumerate}
    \item If $x_2^{(i)}=0$ for all $1\le i\le n-1$, then for any $u\in\mathbb{Q}_p^{n-1}$ and $M\in\mathbb{Z}$,
    \[
    N_{u,M}(x_2)=|\ell_u(x_2)|_p=|x_2^{(n)}|_p.
    \]
    In this case, we may take $c_2=c_1$, $T_2=T_1$, and $L_2=L_1$.
    
    \item Otherwise, we can choose $c_2\in\mathbb{Q}_p^{n-1}$ such that
    \[
    c_2\equiv c_1\ (\mathrm{mod}\ p^{T_1})
    \quad\text{and}\quad
    \ell_{c_2}(x_2)\ne 0.
    \]
    Put
    \[
    \gamma_2=v_p\bigl(\ell_{c_2}(x_2)\bigr)\in\mathbb{Z},
    \qquad
    \beta_2=\min_{1\le i\le n-1}v_p(x_2^{(i)})\in\mathbb{Z},
    \]
    and define
    \[
    T_2=\max\{T_1,\gamma_2-\beta_2+1\},
    \qquad
    L_2=\max\{L_1,\gamma_2-\beta_2\}.
    \]
    For any $(u,M)\in R(c_2,T_2,L_2)$, consider the $p$-adic valuation of $\ell_u(x_2)$. Since
    \[
    \ell_u(x_2)=\ell_{c_2}(x_2)-\sum_{i=1}^{n-1}(u^{(i)}-c_2^{(i)})x_2^{(i)}
    \]
    and
    \begin{align*}
    v_p\Bigl(\sum_{i=1}^{n-1}(u^{(i)}-c_2^{(i)})x_2^{(i)}\Bigr)
    &\ge \min_{1\le i\le n-1}\Bigl(v_p(u^{(i)}-c_2^{(i)})+v_p(x_2^{(i)})\Bigr)\\
    &\ge \min_{1\le i\le n-1}v_p(u^{(i)}-c_2^{(i)})+\min_{1\le i\le n-1}v_p(x_2^{(i)})\\
    &\ge T_2+\beta_2\\
    &\ge \gamma_2+1,
    \end{align*}
    we obtain
    \[
    v_p\bigl(\ell_u(x_2)\bigr)
    =\min\Bigl\{\gamma_2,\,v_p\Bigl(\sum_{i=1}^{n-1}(u^{(i)}-c_2^{(i)})x_2^{(i)}\Bigr)\Bigr\}
    =\gamma_2,
    \]
    and hence
    \[
    |\ell_u(x_2)|_p=p^{-\gamma_2}.
    \]
    Since $M\ge L_2\ge \gamma_2-\beta_2$, we have
    \[
    \frac{\max\limits_{1\le i\le n-1}|x_2^{(i)}|_p}{p^M}
    =p^{-\beta_2-M}
    \le p^{-\beta_2-(\gamma_2-\beta_2)}
    =p^{-\gamma_2}.
    \]
    Therefore,
    \[
    N_{u,M}(x_2)
    =\max\Bigl\{|\ell_u(x_2)|_p,\frac{\max\limits_{1\le i\le n-1}|x_2^{(i)}|_p}{p^M}\Bigr\}
    =p^{-\gamma_2}
    =|\ell_{c_2}(x_2)|_p.
    \]
    
    Moreover, if $(u,M)\in R(c_2,T_2,L_2)$, then $u\equiv c_2\ (\mathrm{mod}\ p^{T_2})$ and $M\ge L_2$. Since $c_2\equiv c_1\ (\mathrm{mod}\ p^{T_1})$ and $T_2\ge T_1$, it follows that $u\equiv c_1\ (\mathrm{mod}\ p^{T_1})$. Also, $L_2\ge L_1$ implies $M\ge L_1$. Hence $(u,M)\in R(c_1,T_1,L_1)$, and therefore
    \[
    R(c_2,T_2,L_2)\subseteq R(c_1,T_1,L_1).
    \]
    \end{enumerate}
    This proves the claim.
    
    Thus, for any $(u,M)\in R(c_2,T_2,L_2)$, the first norm query always returns
    \[
    y_1:=|\ell_{c_1}(x_1)|_p,
    \]
    and the second norm query always returns
    \[
    y_2:=|\ell_{c_2}(x_2)|_p.
    \]
    
    Continuing inductively, we obtain vectors $x_1,x_2,\dots,x_K\in\mathbb{Q}_p^n$, values $y_1,y_2,\dots,y_K\in\mathbb{R}_{\ge 0}$, and triples
    \[
    (c_m,T_m,L_m)\in \mathbb{Q}_p^{n-1}\times\mathbb{Z}\times\mathbb{Z}
    \qquad (m=1,2,\dots,K)
    \]
    such that for each $m=1,2,\dots,K$ and every $(u,M)\in R(c_m,T_m,L_m)$, the first $m$ norm queries performed by $\mathcal{A}$ are
    \[
    x_1,x_2,\dots,x_m
    \]
    (in this order), and the corresponding returned values are
    \[
    y_1,y_2,\dots,y_m
    \]
    (in this order). Moreover, for each $m=2,\dots,K$, we have
    \[
    R(c_m,T_m,L_m)\subseteq R(c_{m-1},T_{m-1},L_{m-1}).
    \]
    
    Because $\mathcal{A}$ is deterministic, the final output of $\mathcal{A}$ (the orthogonal basis) must be the same for all $(u,M)\in R(c_K,T_K,L_K)$. We denote this common output by
    \[
    \{\tilde{\alpha}_1,\tilde{\alpha}_2,\ldots,\tilde{\alpha}_n\},
    \]
    where
    \[
    \tilde{\alpha}_j=(\tilde{\alpha}_j^{(1)},\ldots,\tilde{\alpha}_j^{(n)})\in\mathbb{Q}_p^n
    \qquad (1\le j\le n).
    \]
    
    To derive a contradiction, it suffices to show that there exists $(u,M)\in R(c_K,T_K,L_K)$ such that $\{\tilde{\alpha}_1,\tilde{\alpha}_2,\ldots,\tilde{\alpha}_n\}$ is not an $N_{u,M}$-orthogonal basis.
    
    Indeed, we can choose $u\in\mathbb{Q}_p^{n-1}$ with $u\equiv c_K\ (\mathrm{mod}\ p^{T_K})$ such that
    \[
    \ell_u(\tilde{\alpha}_1)\ne 0
    \qquad\text{and}\qquad
    \ell_u(\tilde{\alpha}_2)\ne 0.
    \]
    Set
    \[
    z_1:=|\ell_u(\tilde{\alpha}_1)|_p,
    \qquad
    z_2:=|\ell_u(\tilde{\alpha}_2)|_p.
    \]
    Choose $M\ge L_K$ such that
    \begin{equation}\label{eq:small-second-term}
    \frac{\max\Bigl\{\max\limits_{1\le i\le n-1}|\tilde{\alpha}_1^{(i)}|_p,\ \max\limits_{1\le i\le n-1}|\tilde{\alpha}_2^{(i)}|_p\Bigr\}}{p^M}
    <\min\{z_1,z_2\}.
    \end{equation}
    Then
    \[
    N_{u,M}(\tilde{\alpha}_1)=z_1,
    \qquad
    N_{u,M}(\tilde{\alpha}_2)=z_2.
    \]
    
    Consider the vector
    \[
    v:=\ell_u(\tilde{\alpha}_2)\tilde{\alpha}_1-\ell_u(\tilde{\alpha}_1)\tilde{\alpha}_2\in\mathbb{Q}_p^n.
    \]
    If $\{\tilde{\alpha}_1,\tilde{\alpha}_2,\ldots,\tilde{\alpha}_n\}$ were an $N_{u,M}$-orthogonal basis, then by the defining property of an orthogonal basis,
    \begin{align}
    N_{u,M}(v)
    &=\max\Bigl\{|\ell_u(\tilde{\alpha}_2)|_p\,N_{u,M}(\tilde{\alpha}_1),\ |\ell_u(\tilde{\alpha}_1)|_p\,N_{u,M}(\tilde{\alpha}_2)\Bigr\}\notag\\
    &=\max\{z_2z_1,z_1z_2\}=z_1z_2. \label{eq:orth-value}
    \end{align}
    
    On the other hand, since $\ell_u$ is linear on $\mathbb{Q}_p^n$,
    \[
    \ell_u(v)=\ell_u(\tilde{\alpha}_2)\ell_u(\tilde{\alpha}_1)-\ell_u(\tilde{\alpha}_1)\ell_u(\tilde{\alpha}_2)=0.
    \]
    Hence, by the definition of $N_{u,M}$,
    \begin{align*}
    N_{u,M}(v)
    &=\max\Bigl\{|\ell_u(v)|_p,\frac{\max\limits_{1\le i\le n-1}|v^{(i)}|_p}{p^M}\Bigr\}\\
    &=\frac{\max\limits_{1\le i\le n-1}\bigl|\ell_u(\tilde{\alpha}_2)\tilde{\alpha}_1^{(i)}-\ell_u(\tilde{\alpha}_1)\tilde{\alpha}_2^{(i)}\bigr|_p}{p^M}\\
    &\le \frac{\max\Bigl\{|\ell_u(\tilde{\alpha}_2)|_p\cdot \max\limits_{1\le i\le n-1}|\tilde{\alpha}_1^{(i)}|_p,\ |\ell_u(\tilde{\alpha}_1)|_p\cdot \max\limits_{1\le i\le n-1}|\tilde{\alpha}_2^{(i)}|_p\Bigr\}}{p^M}\\
    &< \max\{|\ell_u(\tilde{\alpha}_2)|_p,|\ell_u(\tilde{\alpha}_1)|_p\}\cdot \min\{z_1,z_2\}
    \qquad\text{(by \eqref{eq:small-second-term})}\\
    &=\max\{z_1,z_2\}\cdot \min\{z_1,z_2\}=z_1z_2.
    \end{align*}
    Thus $N_{u,M}(v)<z_1z_2$, which contradicts \eqref{eq:orth-value}.
    
    Therefore, the assumption $K<\infty$ is false, and we conclude that
    \[
    \sup_N g_{\mathcal{A}}(N)=\infty.
    \]
    \end{proof}

\section{Norm-Query Complexity of the LVP}

    We begin by recalling the brute-force idea underlying the LVP algorithm of Deng et al.\ \cite{deng2022some}, adapted to the norm-query setting considered here.

Let $V$ be an $n$-dimensional vector space over $\mathbb{Q}_p$, let $N$ be an ultrametric norm on $V$, let $\alpha_1,\cdots,\alpha_m\in V$ be $\mathbb{Q}_p$-linearly independent vectors with $m\le n$, and let
\[
\mathcal{L}:=\mathcal{L}(\alpha_1,\cdots,\alpha_m)
=
\left\{\sum_{i=1}^m a_i\alpha_i:\ a_i\in\mathbb{Z}_p\right\}
\]
be the lattice generated by $\alpha_1,\cdots,\alpha_m$. For each $i\in\mathbb{N}_+$, write
\[
\lambda_i:=\lambda_i(\mathcal{L}).
\]
Our goal is to find a vector $x\in\mathcal{L}$ such that $N(x)=\lambda_2$.

Choose a set
\[
\mathcal{R}\subset\mathbb{F}_p^m\setminus\{0\}
\]
of representatives for the one-dimensional subspaces of $\mathbb{F}_p^m$. We may choose $\mathcal{R}$ so that
\[
e_1,\dots,e_m\in\mathcal{R},
\]
where $e_i$ is the $i$-th standard basis vector of $\mathbb{F}_p^m$. Then
\[
\#\mathcal{R}
=\#\mathbb{P}^{m-1}(\mathbb{F}_p)=
\frac{p^m-1}{p-1}.
\]

We identify the elements of $\mathbb{F}_p$ with their standard representatives in $\mathbb{Z}_p$, and define
\[
\mathcal{S}
:=
\left\{
\sum_{i=1}^m a_i\alpha_i:
(a_1,\dots,a_m)\in\mathcal{R},\
\frac{\lambda_1}{p}
<
N\Bigl(\sum_{i=1}^m a_i\alpha_i\Bigr)
<
\lambda_1
\right\}.
\]
We claim that
\[
\lambda_2=
\begin{cases}
\displaystyle\max_{x\in\mathcal{S}}N(x),
& \text{if }\mathcal{S}\neq\emptyset,\\[6pt]
\lambda_1/p,
& \text{if }\mathcal{S}=\emptyset.
\end{cases}
\]

Let $x\in\mathcal{L}$ be such that
\[
\frac{\lambda_1}{p}<N(x)<\lambda_1
\]
and write
\[
x=\sum_{i=1}^m b_i\alpha_i,
\qquad b_i\in\mathbb{Z}_p.
\]
Since $N(x)>\lambda_1/p$, we have $x\notin p\mathcal{L}$. Hence
\[
(\overline{b}_1,\dots,\overline{b}_m)\neq 0
\]
in $\mathbb{F}_p^m$. Choose $(a_1,\dots,a_m)\in\mathcal{R}$ representing the same one-dimensional subspace. Then there exists $t\in\mathbb{F}_p^\times$ such that
\[
(\overline{b}_1,\dots,\overline{b}_m)
=
t(a_1,\dots,a_m).
\]
Choosing a lift of $t$ to $\mathbb{Z}_p^\times$, we obtain
\[
x-t\sum_{i=1}^m a_i\alpha_i\in p\mathcal{L}.
\]
Therefore
\[
N\left(
x-t\sum_{i=1}^m a_i\alpha_i
\right)
\le
\frac{\lambda_1}{p}
<
N(x).
\]
By the ultrametric property,
\[
N(x)
=
N\left(t\sum_{i=1}^m a_i\alpha_i\right)
=
N\left(\sum_{i=1}^m a_i\alpha_i\right),
\]
since $|t|_p=1$. Hence every norm value in the interval $(\lambda_1/p,\lambda_1)$ attained by a vector in $\mathcal{L}$ is also attained by a vector in $\mathcal{S}$.

Suppose first that $\mathcal{S}\neq\emptyset$. If $x\in\mathcal{L}$ satisfies $N(x)<\lambda_1$, then either
\[
N(x)\le\frac{\lambda_1}{p},
\]
or $N(x)>\lambda_1/p$, in which case the argument above shows that there exists $y\in\mathcal{S}$ such that
\[
N(x)=N(y).
\]
Since every element of $\mathcal{S}$ has norm strictly larger than $\lambda_1/p$, it follows that
\[
\lambda_2=\max_{x\in\mathcal{S}}N(x).
\]

Now suppose that $\mathcal{S}=\emptyset$. Then no vector $x\in\mathcal{L}$ satisfies
\[
\frac{\lambda_1}{p}<N(x)<\lambda_1.
\]
Hence every $x\in\mathcal{L}$ with $N(x)<\lambda_1$ satisfies
\[
N(x)\le\frac{\lambda_1}{p}.
\]
On the other hand, choose $i$ such that $N(\alpha_i)=\lambda_1$. Then
\[
N(p\alpha_i)=\frac{\lambda_1}{p}.
\]
Therefore
\[
\lambda_2=\frac{\lambda_1}{p}.
\]
This proves the claim.

    From the above discussion, we obtain the following search algorithm for finding $x\in\mathcal{L}$ such that $N(x)=\lambda_2$.
    
    \begin{enumerate}
        \item Compute
        \[
        N\Bigl(\sum_{i=1}^m a_i\alpha_i\Bigr)
        \]
        for every $(a_1,\dots,a_m)\in\mathcal{R}$. Since $e_1,\dots,e_m\in\mathcal{R}$, these same queries also give $N(\alpha_1),\dots,N(\alpha_m)$, and hence determine
        \[
        \lambda_1=\max_{1\le i\le m}N(\alpha_i)
        \]
        without any additional norm query.
        \item Among the queried vectors, consider those satisfying
        \[
        \frac{\lambda_1}{p}<N\Bigl(\sum_{i=1}^m a_i\alpha_i\Bigr)<\lambda_1.
        \]
        If such vectors exist, output one having maximal norm.
        \item If no such vector exists, choose any $\alpha_i$ with $N(\alpha_i)=\lambda_1$ and output $p\alpha_i$.
    \end{enumerate}
    
    By construction, the algorithm performs exactly
    \[
    \#\mathcal{R}=\frac{p^m-1}{p-1}
    \]
    norm queries in total, including the queries needed to determine $\lambda_1$.
    
    For $m=1$, one may directly output $p\alpha_1$ without any norm query. Thus the nontrivial case for norm-query complexity is $m\ge 2$. In the next theorem, we will show that, for $m\ge 2$, the above query bound is optimal in the worst case (among deterministic norm-query algorithms).
    
    \begin{theorem}
    Let $V$ be an $n$-dimensional vector space over $\mathbb{Q}_p$, let $\alpha_1,\cdots,\alpha_m\in V$ be $\mathbb{Q}_p$-linearly independent vectors with $2\le m\le n$, and let
    \[
    \mathcal{L}:=\mathcal{L}(\alpha_1,\cdots,\alpha_m)
    =
    \left\{\sum_{i=1}^m a_i\alpha_i:\ a_i\in\mathbb{Z}_p\right\}
    \]
    be the lattice generated by $\alpha_1,\cdots,\alpha_m$. Assume that $\mathcal{A}$ is a deterministic algorithm which accesses the norm only through evaluations of the form $N(x)$, and such that for any ultrametric norm $N$ on $V$, starting from the $\mathbb{Z}_p$-basis $\{\alpha_1,\dots,\alpha_m\}$ of $\mathcal{L}$, the algorithm $\mathcal{A}$ outputs a vector $x\in\mathcal{L}$ satisfying
    \[
    N(x)=\lambda_2(\mathcal{L})
    \]
    in finitely many steps.
    
    For an ultrametric norm $N$ on $V$, let $g_{\mathcal{A}}(N)$ denote the number of norm queries required by $\mathcal{A}$ to find a vector $x\in\mathcal{L}$ with $N(x)=\lambda_2(\mathcal{L})$. Then
    \[
    \sup_N g_{\mathcal{A}}(N)\ge\frac{p^m-1}{p-1},
    \]
    where the supremum is taken over all ultrametric norms $N$ on $V$.
    \end{theorem}

\begin{proof}
Let
\[
\mathcal{P}:=\mathbb{P}(\mathcal{L}/p\mathcal{L}),
\]
the set of one-dimensional $\mathbb{F}_p$-subspaces of
$\mathcal{L}/p\mathcal{L}$. Since
\[
\dim_{\mathbb{F}_p}(\mathcal{L}/p\mathcal{L})=m,
\]
we have
\[
\#\mathcal{P}=\frac{p^m-1}{p-1}.
\]

We prove the theorem by contradiction. Assume that
\[
K:=\sup_N g_{\mathcal{A}}(N)<\#\mathcal{P}.
\]
As in the proof of Theorem~\ref{thm:ortho}, we may assume without loss of
generality that, for every norm, the algorithm makes exactly $K$ norm
queries before it terminates.

Extend $\alpha_1,\dots,\alpha_m$ to a basis
\[
\alpha_1,\dots,\alpha_m,\alpha_{m+1},\dots,\alpha_n
\]
of $V$ over $\mathbb{Q}_p$. When several norms are considered
simultaneously, we write $\lambda_i^{(N)}(\mathcal{L})$ for the lattice
norm values associated with the norm $N$.

With respect to the above basis, define
\[
N_0\Bigl(\sum_{i=1}^n x_i\alpha_i\Bigr)
:=
\max_{1\le i\le n}|x_i|_p.
\]
Then
\[
\lambda_1^{(N_0)}(\mathcal{L})=1,
\qquad
\lambda_2^{(N_0)}(\mathcal{L})=\frac1p.
\]

For each $\xi\in\mathcal{P}$, choose a vector
$\beta_{\xi,1}\in\mathcal{L}$ whose image in
$\mathcal{L}/p\mathcal{L}$ spans $\xi$. Extend the image of
$\beta_{\xi,1}$ to an $\mathbb{F}_p$-basis of
$\mathcal{L}/p\mathcal{L}$, and choose lifts
$\beta_{\xi,2},\dots,\beta_{\xi,m}\in\mathcal{L}$ of the remaining
basis vectors. By Nakayama's lemma,
\[
\beta_{\xi,1},\dots,\beta_{\xi,m}
\]
form a $\mathbb{Z}_p$-basis of $\mathcal{L}$. For $m<j\le n$, put $\beta_{\xi,j}:=\alpha_j$. Then
\[
\beta_{\xi,1},\dots,\beta_{\xi,n}
\]
form a basis of $V$ over $\mathbb{Q}_p$. 

Define
\[
N_\xi\Bigl(\sum_{j=1}^n b_j\beta_{\xi,j}\Bigr)
:=
\max\left\{
p^{-1/2}|b_1|_p,\,
\max_{2\le j\le n}|b_j|_p
\right\}.
\]
It follows immediately from the ultrametric inequality for
$|\cdot|_p$ that $N_\xi$ is an ultrametric norm on $V$.

Since $\beta_{\xi,1},\dots,\beta_{\xi,m}$ form a
$\mathbb{Z}_p$-basis of $\mathcal{L}$, every $y\in\mathcal{L}$ can be
written as
\[
y=\sum_{j=1}^m b_j\beta_{\xi,j},
\qquad
b_j\in\mathbb{Z}_p.
\]
Hence
\[
N_\xi(y)\le 1
\qquad
\text{for every }y\in\mathcal{L}.
\]

We claim that
\[
\lambda_1^{(N_\xi)}(\mathcal{L})=1,
\qquad
\lambda_2^{(N_\xi)}(\mathcal{L})=p^{-1/2}
\]
for every $\xi\in\mathcal{P}$. Indeed, the images of $\alpha_1,\dots,\alpha_m$ form an
$\mathbb{F}_p$-basis of $\mathcal{L}/p\mathcal{L}$. Hence the
one-dimensional subspaces spanned by these images are pairwise distinct.
Since $m\ge2$, there exists some $\alpha_i$ such that the
one-dimensional subspace spanned by its image is different from $\xi$.
Write
\[
\alpha_i=\sum_{j=1}^m b_j\beta_{\xi,j},
\qquad
b_j\in\mathbb{Z}_p.
\]
Since the image of $\alpha_i$ does not lie in the one-dimensional
subspace spanned by the image of $\beta_{\xi,1}$, at least one of
$b_2,\dots,b_m$ is a unit in $\mathbb{Z}_p$. Hence
\[
N_\xi(\alpha_i)=1.
\]
Together with $N_\xi(y)\le1$ for every $y\in\mathcal{L}$, this gives
\[
\lambda_1^{(N_\xi)}(\mathcal{L})=1.
\]

Now let
\[
y=\sum_{j=1}^m b_j\beta_{\xi,j}\in\mathcal{L}
\]
satisfy $N_\xi(y)<1$. By the definition of $N_\xi$, we have $\max_{2\le i\le n}|b_j|_p<1$, which implies that 
\[
b_2,\dots,b_m\in p\mathbb{Z}_p.
\]
Therefore
\[
N_\xi(y)
\le
\max\{p^{-1/2},p^{-1}\}
=
p^{-1/2},
\]
which implies that
\[
\lambda_2^{(N_\xi)}(\mathcal{L})\le p^{-1/2}.
\]
Since
\[
N_\xi(\beta_{\xi,1})=p^{-1/2},
\]
we obtain
\[
\lambda_2^{(N_\xi)}(\mathcal{L})=p^{-1/2}.
\]
This proves the claim.

We next show that a single norm query can distinguish $N_0$ from at
most one norm in the family $\{N_\xi\}_{\xi\in\mathcal{P}}$. Let
\[
x=\sum_{i=1}^n x_i\alpha_i\in V\setminus\{0\},
\qquad
x_i\in\mathbb{Q}_p,
\]
and put
\[
s:=\min_{1\le i\le n}v_p(x_i),
\qquad
\widehat{x}:=p^{-s}x
=
\sum_{i=1}^n\widehat{x}_i\alpha_i.
\]
Then
\[
\widehat{x}_1,\dots,\widehat{x}_n\in\mathbb{Z}_p,
\]
at least one of them is a unit, and
\[
N_0(\widehat{x})=1.
\]

Suppose first that
\[
\widehat{x}_{m+1},\dots,\widehat{x}_n\in p\mathbb{Z}_p.
\]
Then at least one of $\widehat{x}_1,\dots,\widehat{x}_m$ is a unit.
Put
\[
x_{\mathcal{L}}
:=
\sum_{i=1}^m\widehat{x}_i\alpha_i\in\mathcal{L}.
\]
The image $\overline{x_{\mathcal{L}}}$ of $x_{\mathcal{L}}$ in
$\mathcal{L}/p\mathcal{L}$ is therefore nonzero, and we define
\[
\pi(x)
:=
\mathbb{F}_p\cdot\overline{x_{\mathcal{L}}}
\in\mathcal{P}.
\]
If
\[
\widehat{x}_{m+1},\dots,\widehat{x}_n
\]
are not all contained in $p\mathbb{Z}_p$, we leave $\pi(x)$ undefined.

We claim that, for every $\xi\in\mathcal{P}$,
\[
N_\xi(x)
=
\begin{cases}
p^{-1/2}N_0(x), & \text{if }\pi(x)=\xi,\\[4pt]
N_0(x), & \text{otherwise}.
\end{cases}
\]
Since
\[
N_0(x)=p^{-s}
\]
and
\[
N_\xi(x)
=
N_\xi(p^s\widehat{x})
=
p^{-s}N_\xi(\widehat{x})
=
N_0(x)N_\xi(\widehat{x}),
\]
it suffices to prove that
\[
N_\xi(\widehat{x})
=
\begin{cases}
p^{-1/2}, & \text{if }\pi(x)=\xi,\\[4pt]
1, & \text{otherwise}.
\end{cases}
\]

We first consider the case where $\pi(x)$ is defined. Write
\[
x_{\mathcal{L}}
=
\sum_{j=1}^m b_j\beta_{\xi,j},
\qquad
b_j\in\mathbb{Z}_p.
\]
If $\pi(x)=\xi$, then the nonzero image of $x_{\mathcal{L}}$ in
$\mathcal{L}/p\mathcal{L}$ lies in the one-dimensional subspace
spanned by the image of $\beta_{\xi,1}$. Since the images of
$\beta_{\xi,1},\dots,\beta_{\xi,m}$ form an $\mathbb{F}_p$-basis of
$\mathcal{L}/p\mathcal{L}$, it follows that
\[
b_1\in\mathbb{Z}_p^\times,
\qquad
b_2,\dots,b_m\in p\mathbb{Z}_p.
\]
Moreover, since $\pi(x)$ is defined,
\[
\widehat{x}_{m+1},\dots,\widehat{x}_n\in p\mathbb{Z}_p.
\]
Since $\beta_{\xi,j}=\alpha_j$ for $j>m$, these are precisely the
remaining coordinates of $\widehat{x}$ with respect to the basis
$\beta_{\xi,1},\dots,\beta_{\xi,n}$. Hence
\[
N_\xi(\widehat{x})=p^{-1/2}.
\]
If $\pi(x)\ne\xi$, then the image of $x_{\mathcal{L}}$ does not lie in
the one-dimensional subspace spanned by the image of
$\beta_{\xi,1}$. Hence at least one of $b_2,\dots,b_m$ is a unit in
$\mathbb{Z}_p$, and therefore
\[
N_\xi(\widehat{x})=1.
\]

Finally, suppose that $\pi(x)$ is undefined. Then at least one of
\[
\widehat{x}_{m+1},\dots,\widehat{x}_n
\]
is a unit. Since $\beta_{\xi,j}=\alpha_j$ for $j>m$, at least one
coordinate of $\widehat{x}$ corresponding to
$\beta_{\xi,m+1},\dots,\beta_{\xi,n}$ is a unit. Hence
\[
N_\xi(\widehat{x})=1.
\]
This proves the claim. In particular, for any nonzero query vector $x$, there is at most one
$\xi\in\mathcal{P}$ such that
\[
N_\xi(x)\ne N_0(x).
\]
The same conclusion is trivially true for $x=0$. Equivalently, each norm query can eliminate at most one element of
$\mathcal{P}$ from being indistinguishable from $N_0$.

Put
\[
\mathcal{P}_0:=\mathcal{P}.
\]
For each $j=1,\dots,K$, let $x_j$ be the vector used in the $j$-th norm query when $\mathcal{A}$ is run with the norm $N_0$, and put
\[
y_j:=N_0(x_j).
\]
We define subsets
\[
\mathcal{P}_0\supseteq\mathcal{P}_1\supseteq\cdots\supseteq\mathcal{P}_K
\]
recursively by
\[
\mathcal{P}_j
:=
\{\xi\in\mathcal{P}_{j-1}:N_\xi(x_j)=N_0(x_j)\}.
\]

We claim that, for each $j=0,1,\dots,K$ and every
$\xi\in\mathcal{P}_j$, the first $j$ norm queries performed by
$\mathcal{A}$ under $N_\xi$ are the same as those performed under
$N_0$, with the same returned values. Moreover,
\[
|\mathcal{P}_j|
\ge
|\mathcal{P}|-j.
\]

We prove the claim by induction on $j$. The case $j=0$ is immediate.
Suppose that the claim holds for some $j-1$ with $1\le j\le K$.
For every $\xi\in\mathcal{P}_{j-1}$, the first $j-1$ norm queries
and their returned values under $N_\xi$ are identical to those under
$N_0$. Since $\mathcal{A}$ is deterministic, the vector used in the
$j$-th norm query is therefore the same, namely $x_j$. By the
preceding claim, there is at most one $\xi\in\mathcal{P}$ such that
\[
N_\xi(x_j)\ne N_0(x_j).
\]
Hence
\[
|\mathcal{P}_j|
\ge
|\mathcal{P}_{j-1}|-1
\ge
|\mathcal{P}|-j.
\]
For every $\xi\in\mathcal{P}_j$, we also have
\[
N_\xi(x_j)=N_0(x_j)=y_j.
\]
Thus the first $j$ norm queries and their returned values are identical
under $N_\xi$ and $N_0$. This proves the claim.

Since
\[
K<|\mathcal{P}|,
\]
we have
\[
|\mathcal{P}_K|
\ge
|\mathcal{P}|-K
>0.
\]
Choose $\xi\in\mathcal{P}_K$. By the claim, the entire sequence of norm
queries and returned values under $N_\xi$ is identical to that under
$N_0$. Since $\mathcal{A}$ is deterministic, its final output must also
be the same for the two norms. Denote this common output by
\[
x^\ast\in\mathcal{L}.
\]

Since $\mathcal{A}$ is correct for $N_0$, we have
\[
N_0(x^\ast)
=
\lambda_2^{(N_0)}(\mathcal{L})
=
\frac1p.
\]
Since $x^\ast\in\mathcal{L}$, it follows from the definition of $N_0$
that
\[
x^\ast\in p\mathcal{L}.
\]
Hence
\[
x^\ast=py
\]
for some $y\in\mathcal{L}$. Since $N_\xi(y)\le1$, we obtain
\[
N_\xi(x^\ast)
=
\frac1pN_\xi(y)
\le
\frac1p
<
p^{-1/2}
=
\lambda_2^{(N_\xi)}(\mathcal{L}).
\]
Thus $x^\ast$ is not a valid LVP output for $N_\xi$, contradicting
the correctness of $\mathcal{A}$.

Therefore, the assumption
\[
K<\#\mathcal{P}
\]
is false, and we conclude that
\[
\sup_N g_{\mathcal{A}}(N)
\ge
\#\mathcal{P}
=
\frac{p^m-1}{p-1}.
\]
\end{proof}

Together with the search algorithm above, which uses exactly
\[
\frac{p^m-1}{p-1}
\]
norm queries, the preceding theorem yields the following exact complexity result.

    \begin{corollary}
    For a rank-$m$ $p$-adic lattice with $m\ge 2$, the deterministic worst-case norm-query complexity of the LVP is
    \[
   \dfrac{p^m-1}{p-1}.
    \]
    \end{corollary}

\section{Norm-Query Complexity of the CVP}

We now turn to the Closest Vector Problem. We first note that there are
two trivial cases in which no norm query is needed.

If $t\in\mathcal{L}$, then one may simply output $t$, since
\[
N(t-t)=0.
\]
Suppose next that $\mathcal{L}=\mathcal{L}(\alpha_1)$ has rank $1$ and $t\in\mathbb{Q}_p\alpha_1$. Write
\[
t=c\alpha_1,
\qquad
c\in\mathbb{Q}_p.
\]
If $c\in\mathbb{Z}_p$, then $t\in\mathcal{L}$ and we are in the
previous case. If $c\notin\mathbb{Z}_p$, then $|c|_p>1$, and hence for
every $a\in\mathbb{Z}_p$,
\[
|a-c|_p=|c|_p.
\]
Therefore
\[
N(a\alpha_1-t)
=
|a-c|_pN(\alpha_1)
=
|c|_pN(\alpha_1)
\]
for every $a\in\mathbb{Z}_p$. Thus every vector in $\mathcal{L}$ is a
closest vector to $t$, and again no norm query is needed.

Apart from these cases, however, there is no uniform finite bound on
the number of norm queries required to solve the CVP.

\begin{theorem}
Let $V$ be an $n$-dimensional vector space over $\mathbb{Q}_p$ with
$n\ge 2$, let $\alpha_1,\dots,\alpha_m\in V$ be
$\mathbb{Q}_p$-linearly independent vectors with $1\le m\le n$, and let
\[
\mathcal{L}:=\mathcal{L}(\alpha_1,\dots,\alpha_m)
=
\left\{
\sum_{i=1}^m a_i\alpha_i:\ a_i\in\mathbb{Z}_p
\right\}.
\]
Let $t\in V$ satisfy $t\notin\mathcal{L}$ and
$$\dim_{\mathbb{Q}_p}\Big(\mathrm{span}_{\mathbb{Q}_p}\big(\mathcal{L}\cup\{t\}\big)\Big)\ge 2.$$

Assume that $\mathcal{A}$ is a deterministic algorithm which accesses
the norm only through evaluations of the form $N(x)$, and such that,
for any ultrametric norm $N$ on $V$, starting from the
$\mathbb{Z}_p$-basis $\{\alpha_1,\dots,\alpha_m\}$ of $\mathcal{L}$
and the target vector $t$, the algorithm $\mathcal{A}$ outputs a vector
$v\in\mathcal{L}$ satisfying
\[
N(v-t)
=
\min_{w\in\mathcal{L}}N(w-t)
\]
in finitely many steps.

For an ultrametric norm $N$ on $V$, let $g_{\mathcal{A}}(N)$ denote
the number of norm queries required by $\mathcal{A}$ to solve the CVP.
Then
\[
\sup_N g_{\mathcal{A}}(N)=\infty,
\]
where the supremum is taken over all ultrametric norms $N$ on $V$.
\end{theorem}

\begin{proof}
We prove the theorem by contradiction. Assume that
\[
K:=\sup_N g_{\mathcal{A}}(N)<\infty.
\]
As in the proofs of the previous theorems, we
may assume without loss of generality that, for every norm, the
algorithm makes exactly $K$ norm queries before it terminates.

Choose a positive integer $M$ sufficiently large so that
\[
p^M>(M+1)^K,
\]
and let $\mathcal{U}_M\subset\mathbb{Z}_p$ be a complete set of
representatives for
\[
\mathbb{Z}_p/p^M\mathbb{Z}_p.
\]
Thus
\[
\#\mathcal{U}_M=p^M.
\]

We construct a family of ultrametric norms
\[
\{N_{u,M}:u\in\mathcal{U}_M\}
\]
on $V$ for which the corresponding sets of closest vectors are pairwise
disjoint. There are two cases to consider.

\begin{enumerate}
\item Suppose first that
\[
t\notin
\operatorname{span}_{\mathbb{Q}_p}
\{\alpha_1,\dots,\alpha_m\}.
\]
Then $\alpha_1,\dots,\alpha_m,t$ are linearly independent over
$\mathbb{Q}_p$. Extend them to a basis
\[
\alpha_1,\dots,\alpha_m,t,\alpha_{m+2},\dots,\alpha_n
\]
of $V$ over $\mathbb{Q}_p$.

For each $u\in\mathcal{U}_M$, define
\[
N_{u,M}\left(
\sum_{i=1}^m x_i\alpha_i
+
c\,t
+
\sum_{i=m+2}^n x_i\alpha_i
\right)
\]
by
\[
\max\left\{
p^{-M}|c|_p,\,
|x_1+uc|_p,\,
\max\limits_{2\le i\le m}|x_i|_p,\,
\max\limits_{m+2\le i\le n}|x_{i}|_p
\right\},
\]
where terms corresponding to empty index ranges are omitted. It is
straightforward to verify that $N_{u,M}$ is an ultrametric norm on $V$.

Let
\[
v=\sum_{i=1}^m a_i\alpha_i\in\mathcal{L}.
\]
Then
\[
N_{u,M}(v-t)
=
\max\left\{
p^{-M},\,
|a_1-u|_p,\,
\max\limits_{2\le i\le m}|a_i|_p
\right\}\ge p^{-M}.
\]
It follows that
\[
\min_{v\in\mathcal{L}}N_{u,M}(v-t)=p^{-M},
\]
and $v=\sum_{i=1}^m a_i\alpha_i$ is a closest vector if and only if
\[
a_1\equiv u\pmod{p^M}
\]
and
\[
a_2,\dots,a_m\in p^M\mathbb{Z}_p.
\]
Consequently, if $u,u'\in\mathcal{U}_M$ are distinct, no vector in
$\mathcal{L}$ can be a closest vector for both $N_{u,M}$ and
$N_{u',M}$.

\item Suppose next that
\[
t\in
\operatorname{span}_{\mathbb{Q}_p}
\{\alpha_1,\dots,\alpha_m\}.
\]
Since $t\notin\mathcal{L}$, there exists a smallest positive integer
$r$ such that
\[
p^rt\in\mathcal{L}.
\]
By the minimality of $r$,
\[
p^rt\notin p\mathcal{L}.
\]
Put
\[
\beta_1:=p^rt.
\]
The image of $\beta_1$ in $\mathcal{L}/p\mathcal{L}$ is nonzero.
Extend this image to an $\mathbb{F}_p$-basis of
$\mathcal{L}/p\mathcal{L}$, and choose lifts
$\beta_2,\dots,\beta_m\in\mathcal{L}$ of the remaining basis vectors.
By Nakayama's lemma,
\[
\beta_1,\dots,\beta_m
\]
form a $\mathbb{Z}_p$-basis of $\mathcal{L}$. By the assumption of the
theorem, we have $m\ge2$. Extend this basis to a basis
\[
\beta_1,\dots,\beta_m,\beta_{m+1},\dots,\beta_n
\]
of $V$ over $\mathbb{Q}_p$.

For each $u\in\mathcal{U}_M$, define
\[
N_{u,M}\left(\sum_{i=1}^n x_i\beta_i\right)
:=
\max\left\{
p^{-(M+r)}|x_1|_p,\,
|x_2+p^ru x_1|_p,\,
\max_{3\le i\le n}|x_i|_p
\right\}.
\]
It is straightforward to verify that $N_{u,M}$ is an ultrametric norm
on $V$.

Let
\[
v=\sum_{i=1}^m a_i\beta_i\in\mathcal{L}.
\]
Since
\[
t=p^{-r}\beta_1,
\]
we have
\[
v-t
=
(a_1-p^{-r})\beta_1
+
a_2\beta_2
+\cdots+
a_m\beta_m.
\]
Since
\[
p^ra_1-1\in\mathbb{Z}_p^\times,
\]
we obtain
\[
p^{-(M+r)}|a_1-p^{-r}|_p=p^{-M},
\]
and hence
\[
N_{u,M}(v-t)
=
\max\left\{
p^{-M},\,
|a_2-u(1-p^ra_1)|_p,\,
|a_3|_p,\dots,|a_m|_p
\right\}\ge p^{-M}.
\]
Therefore
\[
\min_{v\in\mathcal{L}}N_{u,M}(v-t)=p^{-M},
\]
and $v=\sum_{i=1}^m a_i\beta_i$ is a closest vector if and only if
\[
a_2\equiv u(1-p^ra_1)\pmod{p^M}
\]
and
\[
a_3,\dots,a_m\in p^M\mathbb{Z}_p.
\]

Suppose that the same vector $v$ is a closest vector for both
$N_{u,M}$ and $N_{u',M}$. Then
\[
(u-u')(1-p^ra_1)\in p^M\mathbb{Z}_p.
\]
Since
\[
1-p^ra_1\in\mathbb{Z}_p^\times,
\]
it follows that
\[
u\equiv u'\pmod{p^M}.
\]
Since $u$ and $u'$ belong to the fixed set of representatives
$\mathcal{U}_M$, we must have $u=u'$. Thus, also in this case, the
sets of closest vectors corresponding to distinct elements of
$\mathcal{U}_M$ are pairwise disjoint.
\end{enumerate}

We next claim that, in either case, for every fixed vector
$x\in V$, the set
\[
\{N_{u,M}(x):u\in\mathcal{U}_M\}
\]
contains at most $M+1$ distinct values.

In the first case, write
\[
x
=
\sum_{i=1}^m x_i\alpha_i
+
c\,t
+
\sum_{i=m+2}^n x_i\alpha_i.
\]
If $c=0$, then by definition $N_{u,M}(x)$ is independent of $u$. If $c\ne0$, put
\[
z:=-\frac{x_1}{c}.
\]
After factoring out $|c|_p$, we may write
\[
N_{u,M}(x)
=
|c|_p
\max\left\{
p^{-M},\,
|u-z|_p,\,
D
\right\}
\]
for some $D\ge0$ independent of $u$.

In the second case, write
\[
x=\sum_{i=1}^n x_i\beta_i.
\]
If $x_1=0$, then $N_{u,M}(x)$ is again independent of $u$. If
$x_1\ne0$, put
\[
z:=-\frac{x_2}{p^rx_1}.
\]
Since
\[
|p^rx_1|_p\,p^{-M}
=
p^{-(M+r)}|x_1|_p,
\]
we may similarly write
\[
N_{u,M}(x)
=
|p^rx_1|_p
\max\left\{
p^{-M},\,
|u-z|_p,\,
D
\right\}
\]
for some $D\ge0$ independent of $u$.

Thus, in either case, the only dependence on $u$ is through
\[
\max\{p^{-M},|u-z|_p\}.
\]
If $z\notin\mathbb{Z}_p$, then
\[
|u-z|_p=|z|_p
\]
for every $u\in\mathbb{Z}_p$, so the value is independent of $u$. If
$z\in\mathbb{Z}_p$, then, as $u$ ranges over $\mathcal{U}_M$,
\[
\max\{p^{-M},|u-z|_p\}
\]
can take only the values
\[
1,p^{-1},\dots,p^{-M}.
\]
Therefore
\[
\#\{N_{u,M}(x):u\in\mathcal{U}_M\}\le M+1.
\]
This proves the claim.

We now use this claim to construct a sequence of subsets of
$\mathcal{U}_M$ on which the norm queries made by $\mathcal{A}$ are
indistinguishable. Put
\[
\mathcal{U}_0:=\mathcal{U}_M.
\]

Suppose that, for some $j$ with $1\le j\le K$, a subset
$\mathcal{U}_{j-1}\subseteq\mathcal{U}_M$ has been chosen such that the
first $j-1$ norm queries performed by $\mathcal{A}$, together with
their returned values, are the same for all
$u\in\mathcal{U}_{j-1}$. Since $\mathcal{A}$ is deterministic, the
vector used in the $j$-th norm query is therefore also the same for
all $u\in\mathcal{U}_{j-1}$. Denote this vector by $x_j$.

By the preceding claim, the values
\[
N_{u,M}(x_j),
\qquad
u\in\mathcal{U}_{j-1},
\]
take at most $M+1$ distinct values. Hence there exists a subset
\[
\mathcal{U}_j\subseteq\mathcal{U}_{j-1}
\]
such that the $j$-th query returns the same value for every
$u\in\mathcal{U}_j$, and
\[
\#\mathcal{U}_j
\ge
\frac{\#\mathcal{U}_{j-1}}{M+1}.
\]

Continuing inductively, we obtain subsets
\[
\mathcal{U}_0
\supseteq
\mathcal{U}_1
\supseteq
\cdots
\supseteq
\mathcal{U}_K
\]
such that all $K$ norm queries and their returned values are identical
for every $u\in\mathcal{U}_K$, and
\[
\#\mathcal{U}_K
\ge
\frac{p^M}{(M+1)^K}
>
1.
\]
Choose two distinct elements
\[
u,u'\in\mathcal{U}_K.
\]
Since $\mathcal{A}$ is deterministic, its final output must be the same
under $N_{u,M}$ and $N_{u',M}$. Denote this common output by
\[
v^\ast\in\mathcal{L}.
\]
However, in either of the two cases considered above, the sets of
closest vectors for $N_{u,M}$ and $N_{u',M}$ are disjoint. Hence
$v^\ast$ cannot be a valid CVP output for both norms, contradicting the
correctness of $\mathcal{A}$.

Therefore, the assumption
\[
K<\infty
\]
is false, and we conclude that
\[
\sup_N g_{\mathcal{A}}(N)=\infty.
\]
\end{proof}

\section{Conclusion}

We have studied three basic computational problems in finite-dimensional ultrametric normed spaces from the viewpoint of deterministic norm-query complexity. For orthogonalization, we proved that no algorithm which works for every ultrametric norm can have a uniform finite query bound depending only on the dimension. For the LVP, we determined the exact worst-case norm-query complexity: for a rank-$m$ lattice with $m\ge 2$, it is
\[
\frac{p^m-1}{p-1}.
\]
For the CVP, we showed that, apart from the trivial cases in which no norm query is needed, the worst-case norm-query complexity is unbounded, even when the lattice and the target vector are fixed.

More broadly, these results illustrate that computational problems in $p$-adic normed spaces can exhibit complexity behavior substantially different from the Euclidean setting. The lower-bound arguments also show how limited information obtained from norm queries can impose intrinsic computational obstructions. It would be interesting to study randomized norm-query algorithms and other computational problems in $p$-adic normed spaces from this viewpoint.

\bibliographystyle{plain}
\bibliography{ref}

@book{Weil1974,
  author    = {Weil, André},
  title     = {Basic Number Theory},
  edition   = {Third},
  publisher = {Springer-Verlag},
  address   = {New York-Berlin},
  year      = {1974}
}

@article{deng2025p,
  title     = {On $p$-adic {G}ram--{S}chmidt Orthogonalization Process},
  author    = {Deng, Yingpu},
  journal   = {Frontiers of Mathematics},
  volume    = {20},
  number    = {2},
  pages     = {299--311},
  year      = {2025},
  publisher = {Springer}
}

@article{deng2022some,
  title     = {On some computational problems in local fields},
  author    = {Deng, Yingpu and Luo, Lixia and Pan, Yanbin and Xiao, Guanju},
  journal   = {Journal of Systems Science and Complexity},
  volume    = {35},
  number    = {3},
  pages     = {1191--1200},
  year      = {2022},
  publisher = {Springer}
}
\end{document}